\documentclass[11pt,twoside]{article}
\usepackage[left=3cm,right=3cm,top=3cm,bottom=3cm]{geometry}
\usepackage[
backend=biber,
style=authoryear-comp,
sorting=nyt
]{biblatex}
\usepackage{amsthm,amssymb,mathrsfs,setspace,pstricks,booktabs,mathtools,amsmath,amsfonts,pdflscape,array,upgreek,rotating,graphicx,caption,esint,subcaption,tgcursor,tcolorbox,tikz,stackengine,float}
\usepackage{graphicx}
\usepackage{doi}
\usepackage{siunitx}  
\usepackage{graphicx,type1cm,eso-pic,color}
\usepackage{amssymb}
\usepackage{amsthm}
\usepackage{amssymb,amsmath,bbm}
\usepackage{graphicx,epsfig}
\usepackage{enumerate}
\usepackage{color}
\usepackage {multicol}
\usepackage{lineno}
\usepackage{multirow}
\usepackage{mathptmx}
\numberwithin{equation}{section}
\newtheorem{theorem}{Theorem}
\newtheorem{lemma}{Lemma}

\newtheorem{remark}{Remark}
\newtheorem{corollary}{Corollary}

\makeatletter
\AddToShipoutPicture{
	\setlength{\@tempdimb}{.5\paperwidth}
	\setlength{\@tempdimc}{.5\paperheight}
	\setlength{\unitlength}{1pt}
	\put(\strip@pt\@tempdimb,\strip@pt\@tempdimc)
}
\makeatother

\usepackage{scalerel,stackengine}
\stackMath
\newcommand\reallywidehat[1]{%
	\savestack{\tmpbox}{\stretchto{%
			\scaleto{%
				\scalerel [\widthof{\ensuremath{#1}}]{\kern-.6pt\bigwedge\kern-.6pt}%
				{\rule[-\textheight/2]{1ex}{\textheight}}
			}{\textheight}%
		}{0.5ex}}%
	\stackon[1pt]{#1}{\tmpbox}%
}
\begin{document}
	\setcounter{page}{1}
	\thispagestyle{empty}
	\markboth{}{}

	\pagestyle{myheadings}
	\markboth{}{ }
	
	\date{}
	
	
	\noindent  
	
	\vspace{.1in}
	
	{\baselineskip 20truept
		
		\begin{center}
			{\Large {\bf Characterization based Goodness-of-Fit for Generalized Pareto Distribution: A Blend of Stein's Identity and Dynamic Survival Extropy}} \footnote{\noindent{\bf $^{\#}$} E-mail: nitin.gupta@maths.iitkgp.ac.in,\\
				{\bf * }  corresponding author E-mail: gauravk@kgpian.iitkgp.ac.in}\\
	\end{center}}
	\vspace{.1in}
	
	\begin{center}
		{\large {\bf Gaurav Kandpal$^1$, Nitin Gupta$^2$}}\\
		\vspace{0.1cm}

        {\large {\it ${}^{1,2}$ Department of Mathematics, Indian Institute of Technology Kharagpur, West Bengal 721302, India }}\\
		\end{center}
	
	\vspace{.1in}
	\baselineskip 12truept
	\begin{center}
		{\bf \large Abstract}\\
	\end{center}
	 
     This paper proposes a goodness of fit test for the generalized Pareto distribution (GPD). Firstly, we provide two characterizations of GPD based on Stein's identity and dynamic survival extropy. These characterizations are used to test GPD separately for the positive and negative shape parameter cases. A Monte Carlo simulation is conducted to provide the critical values and power of the proposed test against a good number of alternatives. Our test is simple to use and it has asymptotic normality and relatively high power, which strengthened the purpose of proposing it. Considering the case of right censored data, we provide the procedure to handle censored case too. A few real-life applications are also included.

\section{Introduction}
  The Pareto distribution is of considerable interest across multiple sectors due to its broad applicability and importance in modeling occurrences characterized by heavy-tailed distributions. The extensive utilization of this concept has garnered significant interest from researchers, resulting in the creation of other variants, including type-I, II, III, IV, and generalized Pareto distributions. \cite{arnold2015pareto} offers an extensive examination of many types of Pareto distributions, clarifying their interrelationships. Pareto distributions are the most frequently used models in the fields of finance, economics, and related disciplines. In reality, the initial Pareto distribution, which was attributed to Pareto, was employed to simulate the distribution of wealth among individuals. Several extended Pareto distributions have been proposed in the literature and have been applied in a wide variety of disciplines since \cite{pareto1964cours}. Although the list of applications is excessively extensive, recent applications have included the following: income modeling (\cite{Bhattacharya1999}); wealth distribution in the Forbes 400 list (\cite{KLASS2006290}); commercial fire loss severity in Taiwan (\cite{lee2012fitting}); and city size distribution in the United States (\cite{IOANNIDES201318}).

Pareto distributions are being employed more frequently to simulate economic and financial issues. Therefore, it is imperative to possess instruments that can evaluate the goodness of fit (GOF) of Pareto distributions. In fact, numerous experiments have been suggested to verify the GOF of Pareto distributions, one can refer to \cite{CHU2019} and the references therein. This paper considers the goodness-of-fit test problem for GPD.

\cite{stein1972} established a moment identity for a random variable with a distribution in the exponential family. Stein's type identification has been thoroughly examined in the statistical literature because of its significance in inference methodologies. Thorough analyses of Stein's type identity relevant to various probability distributions and their corresponding characterizations are available in the publications of \cite{kattu2009}, \cite{kattu2012}, \cite{kattu2016}, and \cite{Anastasiou2023}, among others. \cite{betsch2021fixed} established a fixed point characterization for univariate distributions utilizing Stein's type identity. Betsch and Ebner provided many goodness of tests using this fixed point characterization, one may refer to \cite{betsch2019new}, \cite{ebner_logistic2023} and \cite{ebner2024cauchy}. Motivated by this, we provided our first GOF test for testing the generalized Pareto distribution with positive shape parameter. 

\cite{shannon1948mathematical} introduced the notion of information entropy, which measures the average amount of uncertainty about an occurrence associated with a certain probability distribution. \cite{shannon1948mathematical} defined entropy, respectively, for discrete and continuous random variable as 
	\begin{align}
		{H}(\mathbf{p_N})&=-\sum_{i=1}^{N}p_{i}\log{p_i},\\
	\text{and}\ \ \	{H}(X)&=-\int_{-\infty}^{\infty} f(x) \log f(x)dx,
	\end{align}
	where $\mathbf{p_N}=(p_1,\ldots,p_N)$, and $p_i,\ i=1,2\ldots, N$ denote probability mass function (pmf) of a discrete random variable $X,$ $f(x)$ denotes probability density function (pdf) of an absolutely continuous random variable $X$ and $\log$ denotes logarithm to base $e$. 
	
	An alternative measure of uncertainty termed extropy was proposed by \cite{lad2015extropy}, respectively, for discrete and continuous RV as  
	\begin{align}
		J(\mathbf{p_N})& =-\sum^{N}_{i=1}(1-p_i)\log (1-p_i), \\
	\text{and}\ \ \	J(X)&=\dfrac{-1}{2}\int_{-\infty}^{\infty} f^2(x)dx.
	\end{align}
    \cite{jahanshahi2020cumulative} introduced cumulative residual extropy (CRJ) of random variable $X$ as
\begin{equation}
    \xi J(X)=-\frac{1}{2}\int_{0}^{\infty}\bar{F}^2(x)\mathrm{d}x.
\end{equation}
In a recent study, \cite{nair2020dynamic, AbdulSathar19032021} introduced a novel method to evaluate the residual uncertainty in lifetime random variables. The authors propose a dynamic variant of CRJ, referred to as dynamic survival extropy. This measure is defined by
\begin{equation}
    J_s(X; t) = -\frac{1}{2}\int_t^\infty \left(\frac{\bar{F}(x) }{\bar{F}(t)}\right)^2 \, \mathrm{d}x.
\end{equation}
The dynamic survival extropy of $X$ is, in fact, the CRJ of the random variable $[X-t|X>t]$. \cite{AbdulSathar19032021} provided a characterization of modified GPD based on dynamic survival extropy. We provide another characterization for the GPD, where different values of proportionality constant lead to different versions of GPDs. Our second proposed GOF test for GPD with negative shape parameters is motivated by this characterization.

\subsection{Our contribution}
\begin{itemize}
    \item[(a)] A characterization for GPD is provided based on Stein's type identity. Characterization for univariate distributions having either semi-bounded or bounded support has been proposed by \cite{betsch2021fixed}. Since GPD has a support whose type varies according to the shape parameter $\beta$, we provide a new characterization using the results provided by \cite{betsch2021fixed}.
    \item[(b)] We provide another characterization of GPD based on dynamic survival extropy, which can be seen as a characterization of exponential, uniform, and modified GPD for particular values of proportionality constant $k$. One particular case for the proportionality constant $k=-\frac{1}{2}\left(\frac{1+\beta}{2+\beta}\right)$ is also given by \cite{AbdulSathar19032021} in his paper.
    \item[(c)] We propose a goodness of fit test separately for positive and negative values of $\beta$. For $\beta>0$, we use Stein's identity-based characterization, and for $\beta<0$, we use a dynamic survival extropy-based characterization. Our test is simpler in calculation compared to the classical methods, such as the Kolmogorov-Smirnov test and the Anderson-Darling test. A Monte Carlo simulation study with different sample sizes and shape parameter values shows that it has high power even for small sample sizes.
    \item[(d)] Asymptotic properties of the test have been provided under the condition that a consistent estimator of $\theta$ and $\beta$ will be used for an evaluation of test statistics. Being aware of the natural problem of censored data, we extended the test for right censored data.
    \item[(e)] Some real-life applications using real datasets have been added, which includes the ozone (O$_3$) level data excess over 100$\mu g/m^3$ of Delhi, India for June 2015 to November 2017 and zero crossing hourly mean period (in seconds) of the sea waves in Bilbao buoy, Spain. We use our proposed test to identify whether the excess data follows GPD for either positive or negative beta cases.
    \end{itemize}

The paper is organized as follows: In section 2, we provide two new characterizations of GPD using fixed point characterization provided by \cite{betsch2021fixed} and using dynamic survival extropy by \cite{AbdulSathar19032021}. In section 3, we propose a GOF for GPD based on these new characterizations separately for positive and negative shape parameter values. We also include a table containing critical values using Monte Carlo simulation. In section 4, we derive some asymptotic results for our test statistics. We also propose modified test statistics with its asymptotic properties for both cases for right-censored observations. In Section 5, we include the method provided by \cite{VILLASENORALVA20093835} to estimate parameter values while estimating our test statistics. We tabulate power of our test against a large number of alternatives for different values of $\beta$. Finally, we illustrate our test procedures using a few real data sets including a new dataset in literature and one already studied dataset in literature.

\section{Characterizations of GPD}
The GPD was first explicitly introduced by \cite{pickands1975statistical} as a distribution of the exceedance. Later, it was found that many distributions used for long-tailed data can be well approximated by a GPD (\cite{Choulakian01112001}). He suggested that a GPD could often be used as a model for data with a long tail when neither a mode nor an infinite density is suggested by the nature of the variables or by the data themselves.

The cumulative distribution function (CDF) of GPD in its general form is 
\begin{equation}
    F(x;\theta,\beta)=1-\left[1+\frac{\beta}{\theta}x\right]^{-\frac{1}{\beta}} \text{ for } \beta\ne 0,
 \label{cdfGPD}
\end{equation}
where $\theta>0$ and $\beta\in \mathbb{R}$ such that $x>0$ for $\beta>0$ and $0<x<-\frac{\theta}{\beta}$ for $\beta <0$. Notice that for $\beta=-1$, $F(x;\theta,\beta)=\frac{x}{\theta}$ which is $\text{Uniform}(0,\theta)$ distribution. The exponential distribution is a limiting case of $F$ when $\beta\to 0$. Also if $X$ follows distribution $F(x;\theta,\beta)$ given by \eqref{cdfGPD}, then 
\begin{equation*}
    Y=-\left(\frac{1}{\beta}\right)\log{\left(1-\frac{\beta}{\theta}X\right)}
\end{equation*}
is an exponential random variable. The probability density function (PDF) of GPD for $\beta\ne 0$ is 
\begin{equation}
    f(x;\theta,\beta)= \frac{1}{\theta}\left[1+\frac{\beta}{\theta}x\right]^{-\frac{1}{\beta}-1},\quad \theta>0.
 \label{pdf_gpd}
\end{equation}
We shall denote a random variable $X$ having CDF $F(x;\theta,\beta)$ and PDF $f(x;\theta,\beta)$ by $\text{GPD}(\theta,\beta)$, that is $X\sim GPD(\theta,\beta)$. The GPD is a generalization of the Pareto distribution (PD). The PD was studied extensively by \cite{arnold2015pareto}, and the estimation problems in PD were considered by \cite{ArnoldPress01121989}. One interesting and useful property of the GPD is that if $X\sim \text{GPD}(\theta,\beta)$, then $X_t=\{X-t|X>t\}$ will be $\text{GPD}(\theta-\beta t,\beta)$ for any $t>0$. This implies that if the model is consistent with a set of data for a given threshold, then it must be consistent with the data for all higher thresholds. 

Having these properties makes GPD a tool to model problems in economics and finance. Hence, it is essential to have tools to check the goodness of fit of GPD. \cite{CHU2019} provided a review on good number of available tests in the literature. However, there is a lack of a uniformly good method to test GPD. Most existing methods either perform well for certain values of $\beta$ and $\theta$ or they become computationally burdensome as the sample size increases and a very few tests has asymptotic normality. This motivates the development of a new goodness of fit for GPD. In this regard, to develop a goodness-of-fit test, we propose two characterizations for GPD in this section, one is based on Stein's type identity, and the other is based on dynamic survival extropy.
\subsection{Characterization based on Stein's type identity}
\cite{betsch2021fixed} constructed characterization identities for a large class of absolutely continuous univariate distributions based on Stein's method. They provide explicit representations through a formula for the density or distribution function for univariate distributions with semi-bounded support and bounded support. The following two lemmas are from \cite{betsch2021fixed}. 
\begin{lemma}[Semi-bounded support]
\normalfont
    The Stein's characterization for semi-bounded support states that a real-valued random variable $X$ has density $f$ with semi-bounded support $[L,\infty)$ and holds the following conditions
    \begin{enumerate}
        \item[(i)] $P(X\in[L,\infty))=1$,
        \item[(ii)] $\mathbb{E}\left(\left|\frac{f'(X)}{f(X)}\right|\right)<\infty$ and 
        \item[(iii)]  $\mathbb{E}\left(\left|X\frac{f'(X)}{f(X)}\right|\right)<\infty$
    \end{enumerate}
    if and only if the distribution function of $X$ has the form
    \begin{equation}
        F(t)=\mathbb{E}\left(-\frac{f'(X)}{f(X)}(\min(X,t)-L)\right),\quad t>L.
        \label{charGPD1}
    \end{equation}
    \label{lemma_gpd1}
\end{lemma}
\begin{lemma}[Bounded support]
\normalfont
    The Stein's characterization for bounded support states that a real-valued random variable $X$ has density $f$ with support $[L,R]$ and holds the following conditions
    \begin{enumerate}
        \item[(i)] $P(X\in[L,R])=1$,
        \item[(ii)] $\mathbb{E}\left(\left|\frac{f'(X)}{f(X)}\right|\right)<\infty$ and 
        \item[(iii)]  $\lim_{x\to R}f(x)$ exists,
    \end{enumerate}
    if and only if the distribution function of $X$ has the form
    \begin{equation}
        F(t)=\mathbb{E}\left(-\frac{f'(X)}{f(X)}(\min(X,t)-L)\right)+(t-L)\lim_{x\to R}f(x),\quad L<t<R.
        \label{charGPD2}
    \end{equation}
    \label{lemma_gpd2}
\end{lemma}
Since we know that the GPD has a support whose type varies based on the positive and negative values of the shape parameter $\beta$. Therefore, we construct a new characterization for GPD which is based on the methodology of \cite{betsch2021fixed}. Using the above lemmas, since GPD has properties $ (i), (ii)$, and $(iii)$ in both cases, we provide the following characterization for GPD having a support
whose type varies.
\begin{theorem}
    Let $X$ be a positive valued random variable with CDF $F$, then $X$ has generalized Pareto distribution with PDF \eqref{pdf_gpd} if and only if the distribution function of $X$ has a form
\begin{equation}
        F(t)=\mathbb{E}\left(\frac{\beta+1}{\theta+\beta X}\min\{X,t\}\right),\quad t>0.
        \label{char_form_cdf}
\end{equation}
\end{theorem}
\begin{proof}\normalfont
    The support of GPD varies as per the values of $\beta$. When $\beta>0$, the support of GPD is semi-bounded, that is, $[0,\infty)$, then Lemma \ref{lemma_gpd1} provides the characterization mentioned in the theorem. Otherwise, for $\beta<0$, the support of GPD is $\left[0,-\frac{\theta}{\beta}\right]$. We observe that \eqref{charGPD1} and \eqref{charGPD2} differs only with a limit term, and 
    \begin{equation*}
        \lim_{x\to -\frac{\theta}{\beta}} \frac{1}{\theta}\left[1+\frac{\beta}{\theta}x\right]^{-\frac{1}{\beta}-1}=0.
    \end{equation*}
    Therefore, using lemma \ref{lemma_gpd2}, the distribution function of $X$ has the form as \eqref{char_form_cdf}. The expression \eqref{char_form_cdf} can also be verified using integration by parts. This proves the theorem.
\end{proof}
\subsection{Characterization based on dynamic survival extropy}
\cite{AbdulSathar19032021} introduced dynamic survival extropy. He defined two non-parametric classes of distribution based on the monotonicity properties of the dynamic survival extropy. He provided the following theorem, which connects these classes to the value of the hazard rate function $h_F(t)$. 
\begin{theorem}[\cite{AbdulSathar19032021}]
\normalfont
    The distribution function $F$ is increasing (decreasing) dynamic survival extropy classes if and only if for all $t>0$
    \begin{equation}
        J_s(X;t)\cdot h_F(t)\ge (\le) -\frac{1}{4}.
    \end{equation}
\end{theorem}
Further, he provided a characterization of the exponential distribution and generalized Pareto distribution with a modified density. This can be seen as a particular case of our proposed characterization given in the next theorem. 
\begin{theorem}
\normalfont
    Let $F$ be a distribution function with hazard rate $h_F(t)$, then $F$ is a generalized Pareto distribution if and only if $J_s(X;t)\cdot h_F(t)$ is constant.
    \label{thm_charII}
\end{theorem}
\begin{proof}
    Let $F$ be GPD with CDF \eqref{cdfGPD}, then 
    \begin{align*}
       J_s(X;t)&=\frac{\theta+\beta t}{2(\beta-2)}\\
    &=\frac{1}{2(\beta-2)h_F(t)}. 
    \end{align*}
    Hence, $J_s(X;t)\cdot h_F(t)=k$, where $k$ is the proportionality constant.
    
    Conversely, let $J_s(X;t)\cdot h_F(t)=k$, this implies
    \begin{equation}
        f(x)\int_{x}^{\infty}\bar{F}^2(t)\mathrm{d}t=-2k\bar{F}^3(x).
    \end{equation}
    The proof of converse part is similar to the converse part of Theorem 8 in \cite{AbdulSathar19032021}, that is, if $J_s(X;t)\cdot h_F(t)=k$ then the hazard rate function has a form
    \begin{equation}
        h_F(t)=\frac{1}{c_1t+c_2}
        \label{hazard_f}
    \end{equation}
    where $c_1=\dfrac{1+4k}{2k}$ and $c_2=\dfrac{1}{h_F(0)}$, which is hazard rate function for generalized Pareto distribution. In particular for $k=\dfrac{1}{2}\dfrac{1}{\beta-2}$, $h_F(t)=\dfrac{1}{\beta t+\theta}$ which implies $X\sim GPD(\theta,\beta)$ with CDF \eqref{cdfGPD}.
\end{proof}
\begin{remark}
Note that for $k=-\dfrac{1}{4}$, \eqref{hazard_f} is hazard rate function of the exponential distribution and for $k=\dfrac{1}{2}$, \eqref{hazard_f} is hazard rate function of the uniform distribution. For $k=-\dfrac{\beta+1}{2(\beta+2)}$, \eqref{hazard_f} is hazard rate function of the modified Pareto distribution which \cite{AbdulSathar19032021} used in their characterization. 
\end{remark}
\section{Goodness of fit test for testing GPD}
We provided two characterizations of GPD in the previous section. The CDF, we get based on Stein's type identity is not defined at $-\frac{\theta}{\beta}$, so when the shape parameter $\beta<0$, the maximum value of the sample may have value close to $-\frac{\theta}{\beta}$. Such events have also been observed by \cite{Castillo01121997}, when such observation provided senseless results. Therefore, we propose a goodness of fit test separately for $\beta<0$ based on the characterization using dynamic survival extropy. In this way, we have no analytical challenges while applying the proposed test.

Let $X_1, X_2, \ldots, X_n$ be a random sample from the distribution $F$ defined on positive real numbers. We are interested to test the hypothesis,
\begin{align*}
    H_{0} &: F \text{ is GPD}\\
    H_{1}^P &: F \text{ is not GPD when $\beta\ge 0$}\quad \text{or }\\
    H_{1}^N &: F \text{ is not GPD when $\beta<0$}.
\end{align*}
\subsubsection*{Case 1: $\beta\ge0$}
We use Stein's type identity-based characterization in this case and introduce a Cram\'er-von Mises type test statistic as
\begin{equation}
    \Delta_P=\int_{0}^{\infty} \left(\mathbb{E}\left(\frac{\beta+1}{\theta+\beta X}\min(X,t)\right)-F(t)\right)^2\mathrm{d}F(t).
\end{equation}
Under the null hypothesis $H_0$, $\Delta_P$ is zero, whereas under the alternative hypothesis $H_{1}^P$, $\Delta_P$ is non-zero. Further we simplify $\Delta_P$ as
\begin{align}
     \Delta_P=&\int_{0}^{\infty}\mathbb{E}^2\left(\frac{\beta+1}{\theta+\beta X}\min(X,t)\right)\mathrm{d}F(t)\nonumber\\&-2\int_{0}^{\infty}\mathbb{E}\left(\frac{\beta+1}{\theta+\beta X}\min(X,t)\right)F(t)\mathrm{d}F(t) +\int_{0}^{\infty}F^{2}(t)\mathrm{d}F(t)\nonumber\\=&\Delta_a-\Delta_b+\Delta_c\text{ (say)}\label{delta_1}
\end{align}
Consider 
\begin{align}
    \Delta_a&=\int_{0}^{\infty}\mathbb{E}^2\left(\frac{\beta+1}{\theta+\beta X}\min(X,t)\right)\mathrm{d}F(t)\nonumber\\&=(\beta+1)^2 \int_{0}^{\infty}\int_{0}^{\infty}\int_{0}^{\infty}\frac{\min(x,t)\min(y,t)}{(\theta+\beta x)(\theta+\beta y)}\mathrm{d}F(x)\mathrm{d}F(y)\mathrm{d}F(t)\nonumber\\&=(\beta+1)^2 \mathbb{E}\left(\frac{\min(X_1,X_3)\min(X_2,X_3)}{(\theta+\beta X_1)(\theta+\beta X_2)}\right),\label{delta_a}\\
    \Delta_b&=2\int_{0}^{\infty}\mathbb{E}\left(\frac{\beta+1}{\theta+\beta X}\min(X,t)\right)F(t)\mathrm{d}F(t)\nonumber\\&=(\beta+1)\int_{0}^{\infty}\int_{0}^{\infty}\frac{1}{\theta+\beta x}\min(x,t)2F(t)\mathrm{d}F(x)\mathrm{d}F(t)\nonumber\\&=(\beta+1)\mathbb{E}\left(\frac{1}{\theta+\beta X_1}\min(X_1,\max(X_2,X_3))\right)\label{delta_b},
\end{align}
and 
\begin{align}
    \Delta_c&=\int_{0}^{\infty}F^{2}(t)\mathrm{d}F(t)\nonumber\\&=\frac{1}{3}\label{delta_c}.
\end{align}
Substituting \eqref{delta_a}, \eqref{delta_b} and \eqref{delta_c} in \eqref{delta_1}, we get
\begin{align}
    \Delta_P= & (\beta+1)^2 \mathbb{E}\left(\frac{1}{(\theta+\beta X_1)(\theta+\beta X_2)}\min(X_1,X_3)\min(X_2,X_3)\right)\nonumber\\&-(\beta+1)\mathbb{E}\left(\frac{1}{\theta+\beta X_1}\min(X_1,\max(X_2,X_3))\right)\nonumber+\frac{1}{3
    }\nonumber\\=&(\beta+1)^2T_1-(\beta+1)T_2+\frac{1}{3} \text{ (say)}.
\end{align}
Hence, using the theory of $U$-statistics, we consider the $U$-statistics defined by
\begin{align}
    U_p=\binom{n}{3}^{-1}\sum_{1\leq i<j<k\leq n} h_p(X_i,X_j,X_k),\quad p=1,2,
\end{align}
where $h_1$ and $h_2$ are the symmetric kernels defined by 
\begin{align}
    h_1(X_1,X_2,X_3)= \frac{1}{3} &\left( \frac{\min(X_1 , X_3)\min(X_2 , X_3)}{(\theta + \beta X_1)(\theta + \beta X_2)} \right.\nonumber + \frac{\min(X_1 , X_2)\min(X_3 , X_2)}{(\theta + \beta X_1)(\theta + \beta X_3)} \nonumber\\
&+\left. \frac{\min(X_2 , X_1)\min(X_3 , X_1)}{(\theta + \beta X_2)(\theta + \beta X_3)}\right)
\end{align}
and 
\begin{align}
h_2(X_1, X_2, X_3) = \frac{1}{3} &\left(
\frac{\min(X_1 , \max(X_2, X_3))}{\theta + \beta X_1}\right. \nonumber+  \frac{\min(X_2 , \max(X_1, X_3))}{\theta + \beta X_2} \nonumber\\
&+  \left. \frac{\min(X_3 , \max(X_1, X_2))}{\theta + \beta X_3} \right).
\end{align}
Note that $U_1$ and $U_2$ are an unbiased estimator of $T_1$ and $T_2$, respectively. Therefore, the test statistic is given by 
\begin{equation}
    \widehat{\Delta}_P=(\widehat\beta+1 )^2U_1-(\widehat\beta+1)U_2+\frac{1}{3},
\end{equation}
where $\widehat{\beta}$, $\widehat{\theta}$ are consistent estimator of $\beta$ and $\theta$, respectively. The test procedure is to reject the null hypothesis $H_0$ in favor of the alternative hypothesis $H_1^P$ for a large value of $\widehat{\Delta}_P$.

\subsubsection*{Case 2: $\beta<0$}
Define 
\begin{equation}
    \delta(x)=f(x)\int_{x}^{-\frac{\theta}{\beta}}\bar{F}^2(t)\mathrm{d}t+2k\bar{F}^3(x),
\end{equation}
where $k$ is the proportionality constant according to Theorem \ref{thm_charII}. $\delta(x)$ is the measure to study the departure of the dynamic survival extropy of $F$ from the dynamic survival extropy of GPD. Clearly, $\delta(x)=0$ under null hypothesis $H_0$, whereas $\delta(x)\ne 0$ under alternate hypothesis $H_1^N$. Define measure of departure as 
\begin{align*}
    \Delta_N&=\int_{0}^{-\frac{\theta}{\beta}}\delta(x)\mathrm{d}x\\&=\int_{0}^{-\frac{\theta}{\beta}}\int_{x}^{-\frac{\theta}{\beta}}f(x)\bar{F}^2(t)\mathrm{d}t\mathrm{d}x+2k\int_{0}^{-\frac{\theta}{\beta}}\bar{F}^3(x)\mathrm{d}x.
\end{align*}
We use Fubuni's theorem to simplify further and we get
\begin{align*}
    \Delta_N&=\int_{0}^{-\frac{\theta}{\beta}}\bar{F}^2(x)\mathrm{d}x+(2k-1)\int_{0}^{-\frac{\theta}{\beta}}\bar{F}^3(x)\mathrm{d}x\\&=\mathbb{E}[\min(X_1,X_2)]+(2k-1)\mathbb{E}[\min(X_1,X_2,X_3)]\\&=\frac{1}{3}\mathbb{E}\left(\min(X_1,X_2)+\min(X_2,X_3)+\min(X_1,X_3)\right)+(2k-1)\mathbb{E}\left(\min(X_1,X_2,X_3)\right)\\&=\frac{1}{3}R_1+(2k-1)R_2.
\end{align*}
Now using the theory of $U$-statistics, estimator of $R_1$ and $R_2$ are 
\begin{equation}
    W_p=\binom{n}{3}^{-1}\sum_{1\leq i<j<k\leq n}g_p(X_i,X_j,X_k),
\end{equation}
for $p=1$ and 2 respectively. Here 
\begin{align*}
    g_1(X_1,X_2,X_3)&=\frac{1}{3}\left[\min(X_1,X_2)+\min(X_2,X_3)+\min(X_1,X_3)\right]
\end{align*}
and 
\begin{align*}
    g_2(X_1,X_2,X_3)&=\min(X_1,X_2,X_3)
\end{align*}
are symmetric kernels. Therefore, a $U$-statistics based estimator of $\Delta_N$ is 
\begin{equation}
    \widehat{\Delta}_N=\binom{n}{3}^{-1}\sum_{1\leq i<j<k\leq n}g(X_i,X_j,X_k),
\end{equation}
where 
\begin{align*}
    g(X_1,X_2,X_3)&=\frac{1}{3}\left[\min(X_1,X_2)+\min(X_2,X_3)+\min(X_1,X_3)+(6\hat{k}-3)\min(X_1,X_2,X_3)\right]
\end{align*}
is a symmetric kernel and $\hat{k}$ is an estimator of $k$. To make the test scale invariant (since $k$ only depends on $\beta$), we divide $\Delta_N$ by the scale parameter $\theta$ and we get the test statistic as
\begin{equation}
    \widehat{\Delta}_{N}^{*}=\frac{\widehat{\Delta}_{N}}{\widehat{\theta}}.
\end{equation}
Here $\widehat{\theta}$ is a consistent estimator of $\theta$. Therefore, the test is to reject the null hypothesis $H_0$ in favor of the alternative hypothesis $H_1^N$ for larger values of $|\widehat{\Delta}_N^*|$.
\subsection{Table for critical points}
The Monte Carlo method, utilizing 10,000 replications at the $0.05$ and $0.01$ significance levels, is employed to determine the empirical critical values for the proposed test, taking different sample sizes and different shape parameter values.  The parameters of the GPD is derived using the asymptotic maximum likelihood and using the combined estimator of MLE and the moment method based estimator provided by \cite{VILLASENORALVA20093835} for $\beta>0$ and $\beta<0$ respectively. More details regarding estimators are given in Section \ref{estm_par_section}.  To evaluate the critical values of the proposed tests, sample sizes of $n = 20$, 30, 50, 70, and 100 are utilized. In real-life models, it has been observed by \cite{Hosking01081987} that the value of the shape parameter $\beta$ is between $-0.5$ to $0.5$ in most of the cases. So, we include critical values in the range $-1$ to $1$ for $\beta$. 
The critical values are tabulated in Tables \ref{tab:betaN0.01}-\ref{tab:betaP0.05}. All computations and simulations in this paper are conducted solely with R software.
\begin{table}[H]
\centering
\begin{tabular}{m{2cm}  m{2cm} m{2cm}m{2cm}  m{2cm} m{2cm}}
\toprule
$\beta$ & $n=20$ & $n=30$ & $n=50$ & $n=70$ & $n=100$ \\
\midrule
-0.1 & 0.07185 & 0.06257 & 0.05244 & 0.04637 & 0.04109 \\
-0.2 & 0.07063 & 0.05922 & 0.04915 & 0.04276 & 0.03686 \\
-0.3 & 0.06754 & 0.05847 & 0.04531 & 0.04061 & 0.03403 \\
-0.4 & 0.06572 & 0.05477 & 0.04394 & 0.03748 & 0.03275 \\
-0.5 & 0.06346 & 0.05268 & 0.04269 & 0.03476 & 0.02981 \\
-0.6 & 0.06152 & 0.04990 & 0.03990 & 0.03247 & 0.02832 \\
-0.7 & 0.05725 & 0.04806 & 0.03626 & 0.03139 & 0.02555 \\
-0.8 & 0.05554 & 0.04537 & 0.03515 & 0.02927 & 0.02422 \\
-0.9 & 0.05382 & 0.04437 & 0.03318 & 0.02755 & 0.02294 \\
-1.0 & 0.05032 & 0.04052 & 0.03126 & 0.02597 & 0.02178 \\
\bottomrule
\end{tabular}
\caption{Critical values of the test with negative shape parameter $\beta$ at significance level $\alpha=0.01$}
\label{tab:betaN0.01}
\end{table}

\begin{table}[H]
\centering
\begin{tabular}{m{2cm}  m{2cm} m{2cm}m{2cm}  m{2cm} m{2cm}}
\toprule
$\beta$ & $n=20$ & $n=30$ & $n=50$ & $n=70$ & $n=100$ \\
\midrule
-0.1 & 0.05962 & 0.05098 & 0.04240 & 0.03704 & 0.03310 \\
-0.2 & 0.05707 & 0.04757 & 0.03948 & 0.03409 & 0.02941 \\
-0.3 & 0.05345 & 0.04542 & 0.03645 & 0.03206 & 0.02702 \\
-0.4 & 0.05222 & 0.04332 & 0.03465 & 0.02934 & 0.02526 \\
-0.5 & 0.04987 & 0.04025 & 0.03262 & 0.02708 & 0.02327 \\
-0.6 & 0.04733 & 0.03847 & 0.03043 & 0.02526 & 0.02210 \\
-0.7 & 0.04435 & 0.03625 & 0.02805 & 0.02371 & 0.01975 \\
-0.8 & 0.04303 & 0.03496 & 0.02642 & 0.02245 & 0.01859 \\
-0.9 & 0.04072 & 0.03307 & 0.02488 & 0.02111 & 0.01765 \\
-1.0 & 0.03906 & 0.03096 & 0.02361 & 0.01948 & 0.01660 \\
\bottomrule
\end{tabular}
\caption{Critical values of the test with negative shape parameter $\beta$ at significance level $\alpha=0.05$}
\label{tab:betaN0.05}
\end{table}

\begin{table}[H]
\centering
\begin{tabular}{m{2cm}  m{2cm} m{2cm}m{2cm}  m{2cm} m{2cm}}
\toprule
$\beta$ & $n=20$ & $n=30$ & $n=50$ & $n=70$ & $n=100$ \\
\midrule
0.1 & 9.00446 & 1.98907 & 0.79998 & 0.54988 & 0.40142 \\
0.2 & 5.35225 & 1.45293 & 0.57807 & 0.38478 & 0.27909 \\
0.3 & 4.99738 & 1.20743 & 0.41520 & 0.28753 & 0.20763 \\
0.4 & 4.05728 & 0.89690 & 0.32851 & 0.21184 & 0.14626 \\
0.5 & 2.80640 & 0.57507 & 0.23265 & 0.15604 & 0.11385 \\
0.6 & 1.74264 & 0.48629 & 0.17694 & 0.11644 & 0.08622 \\
0.7 & 1.55329 & 0.36463 & 0.12230 & 0.09081 & 0.06837 \\
0.8 & 1.02663 & 0.26791 & 0.10014 & 0.07032 & 0.05422 \\
0.9 & 0.91566 & 0.18860 & 0.08494 & 0.06113 & 0.04750 \\
1.0 & 0.74449 & 0.14882 & 0.07204 & 0.05636 & 0.04269 \\
\bottomrule
\end{tabular}
\caption{Critical values of the test with positive shape parameter $\beta$ at significance level $\alpha=0.01$}
\label{tab:betaP0.01}
\end{table}

\begin{table}[H]
\centering
\begin{tabular}{m{2cm}  m{2cm} m{2cm}m{2cm}  m{2cm} m{2cm}}
\toprule
$\beta$ & $n=20$ & $n=30$ & $n=50$ & $n=70$ & $n=100$ \\
\midrule
0.1 & 2.08490 & 0.88480 & 0.46333 & 0.35089 & 0.28085 \\
0.2 & 1.58367 & 0.62699 & 0.32530 & 0.24698 & 0.19972 \\
0.3 & 1.26252 & 0.46546 & 0.23426 & 0.17948 & 0.14323 \\
0.4 & 0.91523 & 0.34585 & 0.17208 & 0.12947 & 0.10377 \\
0.5 & 0.63462 & 0.24090 & 0.12751 & 0.09719 & 0.07713 \\
0.6 & 0.45782 & 0.18354 & 0.09556 & 0.07328 & 0.05926 \\
0.7 & 0.34632 & 0.13119 & 0.07275 & 0.05464 & 0.04655 \\
0.8 & 0.23656 & 0.10678 & 0.05844 & 0.04509 & 0.03653 \\
0.9 & 0.19128 & 0.07891 & 0.04776 & 0.03730 & 0.03040 \\
1.0 & 0.14527 & 0.06918 & 0.04088 & 0.03228 & 0.02590 \\
\bottomrule
\end{tabular}
\caption{Critical values of the test with positive shape parameter $\beta$ at significance level $\alpha=0.05$}
\label{tab:betaP0.05}
\end{table}

The parametric bootstrap method serves as an effective statistical technique for estimating critical points across different hypothesis testing situations. This method involves making multiple new samples from a fitted parametric model, which helps to carefully examine how the test statistic behaves when the null hypothesis is true. The essential aspect, which plays a crucial role in determining whether to dismiss the null hypothesis, is discerned through this resampling method. The algorithm employed in this study is detailed in Algorithm $A_1$, offering a structured and methodical approach to implementing the parametric bootstrap in practice. The algorithm utilizes the parametric bootstrap technique to estimate the critical value. We generated 10,000 resampled datasets to compute the test statistic for each sample, subsequently deriving critical values $(C1, C2)$ from the empirical method for $(95\%,99\%)$ confidence. The empirical distribution of these statistics. The null hypothesis $H_0$ is rejected when the observed test statistic exceeds these critical values.

\begin{table}[H]
    \begin{tabular}{l}
    \toprule\\
\textbf{Algorithm $A_1$:} \textit{A bootstrap algorithm to find $C_1$ and $C_2$ ($\beta<0$ case)}\\
\midrule\\
$x$ : A numeric vector of data values.\\
$\bar{X} = \text{mean}(x)$\\
$n = \text{length}(x)$.\\
$\beta \longleftarrow \dfrac{\bar{X}}{\bar{X} - \text{max}(x)}$\\
$k\longleftarrow\dfrac{1}{2(\beta-2)}$\\
$\theta \longleftarrow -\beta\cdot\max(x)$\\
$\widehat\Delta_N^*(x, \beta,k,\theta)$.\quad \# define \& compute the test statistic \\
$B \longleftarrow 10000$\\
for$(b \text{ in }1:B)$\{\\
$i \longleftarrow \text{sample}(1:n, \text{size}=n, \text{replace}=TRUE)$\\
$y \longleftarrow x[i]$\\
$\Delta_N[b]\longleftarrow \widehat\Delta_N^*(y, \beta,k,\theta)$\\
\}\\
$\widehat\Delta_N^* \longleftarrow \text{sort}(\Delta_N)$.\\
    $
    C_1 \longleftarrow \text{quantile}(\widehat\Delta_N^*,0.95),$\\
    $C_2 \longleftarrow \text{quantile}(\widehat\Delta_N^*, 0.99)
    $\\
    if ($\widehat\Delta_N^*> C_1$) print(``Reject $H_0$") \text{else print(``Accept $H_0$")}\quad
     \# with $0.05$ level of significance\\

if ($\widehat\Delta_N^* > C_2$) print(``Reject $H_0$") \text{else print(``Accept $H_0$")}\quad
     \# with $0.01$ level of significance\\
    \bottomrule
        
    \end{tabular}
    \label{tab:algo}
\end{table}

\section{Asymptotic properties and test for censored data}
According to \cite{lehmann1951consistency}, $U_1$, $U_2$, $W_1$ and $W_2$ are consistent estimators of $T_1$, $T_2$, $R_1$ and $R_2$, respectively, as they are $U$-statistics. Hence, we obtained
 the following result using the asymptotic theory of $U-$statistics. We denote convergence in probability and convergence in distribution by $\overset{P}{\rightarrow}$ and $\overset{d}{\rightarrow}$, respectively.
 \subsection{Asymptotic properties}
\begin{theorem}
    Let $\widehat{\beta}$ and $\widehat{\theta}$ be the consistent estimators of $\beta$ and $\theta$, respectively. As $n\to\infty$, under $H_1^P$,  $\widehat{\Delta}_P\overset{P}{\rightarrow}\Delta_P$ and under $H_1^N$, $\widehat{\Delta}_N\overset{P}{\rightarrow}\Delta_N$.
\end{theorem}
\begin{theorem}
    Let $\widehat{\beta}$ and $\widehat{\theta}$ be the consistent estimators of $\beta$ and $\theta$, respectively. The distribution of $\sqrt{n}(\widehat{\Delta}_P-\Delta_P)$ converges to a normal random variable with mean zero and variance $9\sigma^2$ as $n\to \infty$, where $\sigma^2$ is obtained by 
    \begin{equation*}
        \sigma^2=Var[\mathbb{E}(h(X_1,X_2,X_3)|X_1)].
    \end{equation*}
    \label{thm:asymPositve}
\end{theorem}
\begin{proof}
    Define 
    \begin{equation*}
        \breve{\Delta}_P=(\beta+1)^2U_1-(\beta+1)U_2+\frac{1}{3}.
    \end{equation*}
    Consider
    \begin{equation*}
        \sqrt{n}(\widehat{\Delta}_P-\Delta_P)=\sqrt{n}(\widehat{\Delta}_P-\breve\Delta_P)+\sqrt{n}(\breve{\Delta}_P-\Delta_P).
    \end{equation*}
    Further, we get
    \begin{equation*}
        \sqrt{n}(\widehat{\Delta}_P-\breve\Delta_P)=\sqrt{n}\left((\widehat\beta+1)^2-(\beta+1)^2\right)U_1-\sqrt{n}\left((\widehat\beta+1)-(\beta+1)\right)U_2. 
    \end{equation*}
    Since $\widehat\beta$ be the consistent estimator of $\beta$. Then 
    \begin{equation}
        \widehat{\beta}\overset{P}{\rightarrow}\beta,\quad U_1\overset{P}{\rightarrow}\mathbb{E}(U_1), \quad U_2\overset{P}{\rightarrow}\mathbb{E}(U_2).
    \end{equation}
    This implies
    \begin{equation*}
        \left((\widehat\beta+1)^2-(\beta+1)^2\right)U_1\overset{P}{\rightarrow}0,\quad \left((\widehat\beta+1)-(\beta+1)\right)U_2\overset{P}{\rightarrow}0.
    \end{equation*}
    Using Chebyshev's inequality, $\sqrt{n}(\widehat{\Delta}_P-\breve{\Delta}_P)\overset{P}{\rightarrow}0$, the central limit theorem of $U-$statistics and the fact that $\mathbb{E}(\breve\Delta_p)=\Delta_P$, we get
    \begin{equation*}
        \sqrt{n}(\breve{\Delta}_P-\Delta_P)\overset{d}{\rightarrow}N(0,9\sigma^2).
    \end{equation*}
    Finally using Slutsky's theorem we get
    \begin{equation*}
        \sqrt{n}(\widehat{\Delta}_P-\Delta_P)\overset{d}{\rightarrow}N(0,9\sigma^2).
    \end{equation*}
    Here $9\sigma^2$ is the asymptotic variance and given by
    \begin{equation}\sigma^2=Var[\mathbb{E}(h(X_1,X_2,X_3)|X_1)],\label{var_postH1}\end{equation}
    where
    \begin{align*}
        h(X_1,X_2,X_3)=&\frac{1}{3}\left((\beta+1)^2\left(\frac{\min(X_1 , X_3)\min(X_2 , X_3)}{(\theta + \beta X_1)(\theta + \beta X_2)} \right.\nonumber + \frac{\min(X_1 , X_2)\min(X_3 , X_2)}{(\theta + \beta X_1)(\theta + \beta X_3)}\right.\\&+\left.\frac{\min(X_1 , X_2)\min(X_3 , X_1)}{(\theta + \beta X_2)(\theta + \beta X_3)}\right)-(\beta+1)\left(
\frac{\min(X_1 , \max(X_2, X_3))}{\theta + \beta X_1}\right. \\&+  \frac{\min(X_2 , \max(X_1, X_3))}{\theta + \beta X_2}
+  \left. \frac{\min(X_3 , \max(X_1, X_2))}{\theta + \beta X_3} \right)\left.+1\right).
    \end{align*}
\end{proof}
Under the null hypothesis $H_0$, $\Delta_P=0$. Hence, the following corollary is obtained.
\begin{corollary}
\label{cor:posBeta}
    Under $H_0$, as $n\to\infty$, $\sqrt{n}\widehat{\Delta}_P$ converges in distribution to a normal random variable with mean zero and variance $9\sigma_0^2$, where $\sigma_0^2$ is obtained by \eqref{var_postH1} evaluating under $H_0$..
\end{corollary}
The asymptotic critical region for the test can be obtained using Corollary \ref{cor:posBeta}. Let $\widehat{\sigma}_0^2$ be a consistent estimator of $\sigma_0^2$, then for the positive $\beta$ case, the null hypothesis $H_0$ is rejected in favor of the alternative hypothesis $H_1^P$ at a significance level of $\alpha$ if
\begin{equation}
    \sqrt{n}\frac{\widehat\Delta_p}{3\widehat{\sigma}_0}>z_{\alpha},
\end{equation}
where $z_{\alpha}$ is the upper $\alpha-$percentile of the standard normal distribution. It is visible that finding null variance $\sigma_0^2$ is a difficult task, so we suggest to obtain critical region of the test using bootstrap procedure. Similarly for negative $\beta$ case, the following theorem can be derived by using the same idea used in Theorem \ref{thm:asymPositve}.
\begin{theorem}
    As $\hat{k}$ and $\widehat{\theta}$ be the consistent estimator of $k$ and $\theta$, respectively. The distribution of $\sqrt{n}(\widehat{\Delta}_N-\Delta_N)$ converges to a normal random variable with mean zero and variance $9\sigma^2$ as $n\to\infty$, where $\sigma^2$ is obtained by 
    \begin{equation}
        \sigma^2=Var[\mathbb{E}(g(X_1,X_2,X_3)|X_1],
        \label{sigma:eval}
    \end{equation}
    where
    \begin{equation*}
        g(X_1,X_2,X_3)=\frac{1}{3}\left[\min(X_1,X_2)+\min(X_2,X_3)+\min(X_1,X_3)+(6\hat{k}-3)\min(X_1,X_2,X_3)\right].
    \end{equation*}
\end{theorem}
\begin{corollary}
    Under $H_0$, as $n\to\infty$, $\sqrt{n}\widehat{\Delta}_N$ converges in distribution to a normal random variable with mean zero and variance $9\sigma_1^2$, where $\sigma_1^2$ is obtained by \eqref{sigma:eval} evaluating under $H_0$.
\end{corollary}
Now, using Slutsky’s theorem, the following result can be obtained using the above corollary.
\begin{corollary}
\label{neg:beta:corr}
    Under $H_0$, as $n\to\infty$, $\sqrt{n}\widehat{\Delta}_N^*$ converges in distribution to a normal random variable with mean zero and variance $\sigma_0=9\dfrac{\sigma_1^2}{\theta^2}$.
\end{corollary}

The asymptotic critical region for the scale invariant test can be obtained using Corollary \ref{neg:beta:corr}. Let $\widehat{\sigma}_0^2$ be a consistent estimator of $\sigma_0^2$, then for the negative $\beta$ case, the null hypothesis $H_0$ is rejected in favor of the alternative hypothesis $H_1^N$ at a significance level of $\alpha$ if
\begin{equation}
    \sqrt{n}\frac{|\widehat\Delta_N^*|}{\widehat{\sigma}_0}>z_{\alpha/2},
\end{equation}
where $z_{\alpha}$ is the upper $\alpha-$percentile of the standard normal distribution.
\subsection{Test for censored observation}
 Occurrences of right-censored observations are frequently seen in the analysis of lifetime data. A very few techniques address the issue of testing for GPD using censored samples. An alternative method involves replacing the distribution function with the Kaplan-Meier estimator in order to calculate the test statistic. In this technique, it is necessary to modify the metric used to quantify deviation from the null hypothesis in the presence of censored observations.
    Another method is the inverse probability censoring weighted scheme (IPCW), in which the censored data is adjusted by weighting it with the inverse of the survival function of the censoring variable provided by \cite{koul1981regression}, \cite{rotnitzky2005inverse} and \cite{datta2010inverse}. In this discussion, we explore the approach to address instances of censorship.

Assume that we have randomly censored observations, meaning that the censoring times are unrelated to the lifetimes and occur independently. Let the observed data are $n$ independent and identical (i.i.d.) copies of $(X^*,\delta)$, with $X^*=\min(X,C)$, where $C$ is the censoring time and $\delta=I(X\leq C)$. We investigate the testing problem mentioned based on $n$ i.i.d. observations $\{(X_i,\delta_i),\, 1\leq i\leq n\}$. Note that $\delta_i =0$ means that the $i$th object is censored by $C$, on the right and $\delta_i =1$ means $i$th object is not censored. We refer to \cite{koul1980testing} to define measure $\Delta_N$ for censored observations. We refer to \cite{datta2010inverse} to get an estimator of $\Delta_N$ ($\beta<0$ case) with censored observation as 
\begin{equation}
     \widehat{\Delta}_N^C =\frac{6}{n(n-1)(n-2)}\sum_{1\leq i<j<k\leq n} \frac{g(X_i^*,X_j^*,X_k^*)\delta_i \delta_j \delta_k}{\widehat{K}_c(X_i^*) \widehat{K}_c(X_j^*)\widehat{K}_c(X_k^*)},
\end{equation}
where $\widehat{K}_c(X_i^*), \widehat{K}_c(X_j^*),\widehat{K}_c(X_k^*)$ are strictly positive with probability one and 
\begin{align*}g(X_1^*,X_2^*,X_3^*)=\frac{1}{3}\left[\min(X_1^*,X_2^*)+\min(X_2^*,X_3^*)+\min(X_1^*,X_3^*)+(6\widehat{k}_c-3)\min(X_1^*,X_2^*,X_3^*)\right].
\end{align*}
Here, $\widehat{K}_c$ is the Kaplan-Meier estimator of $K_c$, the survival function of the censoring variable $C$ and \begin{equation}
    \widehat{k}_c=\frac{1}{2(\hat\beta_c-2)}\quad\text{ and }\quad\widehat{\beta}_c=\frac{\overline{X}_c}{\overline{X}_c-X_{(n)}^c},
\end{equation}
where $X_{(n)}^c=\max{\{X_i^*, 1\le i\le n\}}$ and
\begin{equation}
    \overline{X}_c=\frac{1}{n}\sum_{i=1}^{n}\frac{X_i^* \delta_i}{\widehat{K}_c(X_i^*)}.
\end{equation}
 Since $\overline{X}_c$ and $X_{(n)}^c$ are a consistent estimator of $\overline{X}$ and $X_{(n)}$ for censored observations, therefore using continuous mapping theorem for convergence in probability, we get $\widehat{\beta}_c$ and $\widehat{\theta}_c=-\widehat\beta_c\cdot X_{(n)}^c$ are consistent estimator of $\beta$ and $\theta$ for censored observations. Therefore, in the right censoring situation, the test statistic is given by
\begin{equation}
    \widehat{\Delta}_c^*=\frac{\widehat{\Delta}_N^C}{\widehat{\theta}_c},
\end{equation}
and the test procedure is to reject null hypothesis $H_0$ in favor of $H_1^N$ for larger values of $|\widehat{\Delta}_c^*|$.

For deriving the asymptotic distribution of $\widehat{\Delta}_c^*$, let us define $N_i^c(t) = I(X_i^* \le t, \delta_i = 0)$ as the counting process corresponding to the censoring random variable for the $i$-th subject and $R_i(u) = I(X_i^* \ge u)$. Let $\lambda_c(t)$ be the hazard rate of the censoring variable $C$. The martingale associated with the counting process \(N_i^c(t)\) is given by
\begin{equation}
M_i^c(t) = N_i^c(t) - \int_0^t R_i(u) \lambda_c(u) \, \mathrm{d}u.
\end{equation}

Let $G(x,y)=P(X_1\leq x, X_1^*\leq y, \delta_1=1)$, $x\in \mathbb{R}$, $H(t)=P(X_1^*\geq t)$ and 
\begin{equation}
    w(t) = \frac{1}{H(t)} \int_{\mathbb{R}\times [0,\infty)} \frac{g_c(x)}{K_c(y -)} I(y > t) \, \mathrm{d}G(x,y),
\end{equation}
where $g_c(x)=\mathbb{E}\left[g(x,X_2^*,X_3^*)\right]$. The next theorem follows from \cite{datta2010inverse} for the choice of the kernel 
\begin{align*}g_c(X_1^*,X_2^*,X_3^*)=\frac{1}{3}\left[\min(X_1^*,X_2^*)+\min(X_2^*,X_3^*)+\min(X_1^*,X_3^*)+(6\hat{k}_c-3)\min(X_1^*,X_2^*,X_3^*)\right]
\end{align*}
and under the assumption $\mathbb{E}g_c^2(X_1^*,X_2^*,X_3^*)<\infty$, 
\begin{equation*}
    \int_{\mathbb{R}\times [0,\infty)} \frac{g_c^2(x)}{K_c^2(y)}\, \mathrm{d}G(x,y)<\infty,
\end{equation*}
and 
\begin{eqnarray*}
    \int_{0}^{\infty} w^2(t)\lambda_c(t)\mathrm{d}t<\infty.
\end{eqnarray*}
\begin{theorem}
\label{theorem:censor}
    The distribution of $\sqrt{n}\left(\widehat{\Delta}_N^C-\Delta_N\right)$, as $n\rightarrow \infty$, is Gaussian with mean zero and variance $9\sigma_{1c}^2$, where $\sigma_{1c}^2$ is given by
    \begin{eqnarray*}
        \sigma_{1c}^2=Var\left( \frac{g_c(X)\delta_1}{K_c(X^*)}+\int_{0}^{\infty} w(t)\,\mathrm{d}M_1^c(t)\right).
    \end{eqnarray*}
\end{theorem}
\begin{corollary}
    Under the assumption of Theorem \ref{theorem:censor}, if $\mathbb{E}(X_1^2)<\infty$, the distribution of  $\sqrt{n}\left(\widehat{\Delta}_c^*-\Delta_N^*\right)$, as $n\rightarrow \infty$, is Gaussian with mean zero and variance $9\sigma_{c}^2$, where $\sigma_{c}^2$ is given by
    \begin{equation}
        \sigma_c^2=\frac{\sigma_{1c}^2}{\theta^2},
    \end{equation}
    where $\Delta_n^*=\dfrac{\Delta_N}{\theta}$.
\end{corollary}
\begin{proof}
    The consistency of the estimator $\widehat{\theta}_c$ for $\theta$ is proved using the consistency of $\overline{X}_c$ for $\overline{X}$ given by \cite{zhao2000estimating}. Therefore, the result follows from the above theorem by applying Slutsky's theorem. 
\end{proof}
As suggested by \cite{datta2010inverse}, the reweighted average technique is used to simplify the asymptotic analysis. therefore the reweighted approach is used to find an estimator of $\sigma_{1c}^2$. An estimator of $\sigma_{1c}^2$ is given by
\begin{equation*}
    \widehat{\sigma}_{1c}^2=\frac{9}{n-1}\sum_{i=1}^{n}(V_i-\overline{V})^2,
\end{equation*}
where
\begin{align*}
    &\widehat{h}_1(x)=\frac{1}{n^2}\sum_{1\le j<k\le n}^{n}\frac{g(x,X_j^*,X_k^*)\delta_j\delta_k}{\widehat{K}_c(X_j^*)\widehat{K}_c(X_k^*)},\quad\quad\quad \xi_j
    =\frac{\widehat{h}_1(X_j)\delta_j}{\widehat{K}_c(X_j^*)}\\&\widehat{w}(X_i^*)=\frac{\sum_{j=1}^{n}\xi_jI(X_j^*>X_i^*)}{\sum_{j=1}^{n}I(X_j^*\ge X_i^*)},\quad\quad\quad \quad\quad  \phi_i=\widehat{w}(X_i^*)(1-\delta_i)\\
    &V_i=\xi_i+\phi_i-\sum_{j=1}^n\frac{\phi_iI(X_i^*>X_j^*)}{\sum_{j=1}^{n}I(X_j^*\ge X_i^*)},\quad \text{ and }\quad \overline{V}=\frac{1}{n}\sum_{i=1}^n V_i.
\end{align*}
Therefore, an estimator for $\sigma_c^2$ is given by
\begin{equation}
    \widehat{\sigma}_c^2=\frac{\widehat{\sigma}_{1c}^2}{\widehat{\theta}_c}.
    \label{eqn:censoring:var}
\end{equation}
\begin{corollary}
    Under the assumption in Theorem \ref{theorem:censor}, let $\sigma_{0c}^2$ denote the value of $\sigma_c^2$ when evaluated under $H_0$. As $n\to\infty$, $\sqrt{n}\widehat{\Delta}_c^*$ will converge in distribution to a Gaussian random variable with mean zero and variance $9\sigma_{0c}^2$ under the null hypothesis $H_0$.
\end{corollary}

Therefore, in the case of right censoring, we reject null hypothesis $H_0$ in favor of $H_1^N$ at a significance level $\alpha$, if
\begin{equation}
    \frac{\sqrt{n}|\widehat{\Delta}_c^*|}{3\widehat{\sigma}_{0c}}>z_{\alpha/2},
\end{equation}
where $\widehat{\sigma}_{0c}$ is a consistent estimator of $\sigma_{0c}$ and can be estimated using \eqref{eqn:censoring:var} under $H_0$ and $z_{\alpha}$ is the upper $\alpha-$percentile of the standard normal distribution.

\begin{remark}
    Similarly, we can prove the positive shape parameter case by taking consistent estimator of $\beta$ and $\theta$ for censored observation. The proof will be in similar fashion, hence omitted.
\end{remark}
\section{Simulation study}
This section is divided into three parts. First, we present the method for estimating parameters in our proposed test, as provided by \cite{VILLASENORALVA20093835}, addressing both cases separately. Next, we include a power analysis of our test. Finally, we discuss some real-life applications.

\subsection{Estimation of parameters}
\label{estm_par_section}
The most usual methods for estimating the parameters of a GPD are maximum likelihood (ML), method of moments (MOM) and probability weighted moments approaches. One may refer to \cite{Hosking01081987} for a detailed study and comparison of these three estimator. Let $x_1,x_2,\ldots,x_n$ be $n$ iid realizations of $X$, where $X\sim GPD(\theta,\beta)$. The ML estimation for a GPD parameters $(\theta,\beta)$ are simultaneous solution of 
\begin{equation}
    n\theta-(\beta+1)\sum_{i=1}^{n}\left[1+\frac{x_i}{\theta}\right]^{-1}x_i=0
\end{equation}
and
\begin{equation}
    \theta\sum_{i=1}^{n}\log \left[1+\frac{x_i}{\theta}\right]-(\beta+1)\sum_{i=1}^{n}\left[1+\frac{x_i}{\theta}\right]^{-1}x_i=0.
\end{equation}
When $\beta<-1$, the log-likelihood function can be made as large as possible by taking $\theta$ arbitrary close to $1/x_{(n)}$, where $x_{(n)}=\max{(x_i,i=1,\ldots, n)}$. Therefore in such condition ML estimator do not exists. In addition when $-1<\beta<-0.5$, the ML estimators do not perform well as given by  \cite{grimshaw1993computing}. We also check our test using ML estimator by simulation, which does not provide suitable results. The MOM estimators of $\theta$ and $\beta$ are
\begin{equation}
    \widehat{\beta}_{MOM}=\frac{1}{2}\left(1-\frac{\bar{X}^2}{S^2}\right),
\end{equation}
and
\begin{equation}
    \widehat{\theta}_{MOM}=\frac{1}{2}\bar{X}\left(1+\frac{\bar{X}^2}{S^2}\right).
\end{equation}
where $\bar{X}$ and $S^2$ as sample mean and sample variance respectively. The first two moments of GPD exists only when $\beta<1$ and $\beta<0.5$, respectively. therefore, the we can apply MOM and probability weighted moment estimators to a restricted value of $\beta$. There are some other approaches also using Bayesian perspective by \cite{zhang2009new}, elemental percentile method (EPM) by \cite{Castillo01121997}, minimum distance estimation method by \cite{Chen02102017}. These methods works for all values of $\beta$, but while using these for our proposed test, it didn't work well. Additionally, these method has a big computational cost, which may be too high when sample size will be large. \cite{VILLASENORALVA20093835} provided two new estimators namely asymptotic maximum likelihood (AML) estimators and combination of ML and MOM (CMM). The AML estimators of $\beta$ and $\theta$ are
\begin{equation}
    \widehat{\beta}_{AML}=-W_{n-k+1}+\frac{1}{k}\sum_{j=1}^{k}W_n-j+1
\end{equation}
and
\begin{equation}
    \widehat{\theta}_{AML}=\widehat{\beta}_{AML}\exp\left[W_{n-k+1}+\widehat{\beta}_{AML} \log\left(\frac{k}{n}\right)\right],
\end{equation}
where $W_j=\log X_{(j)}$, $1\le k\le n$ and $X_{(j)}$ is $j$th order statistics. This estimator exists for all values of $k$ and it works well for our test. The CMM estimators of $\beta$ and $\theta$ are
\begin{equation}
    \widehat{\beta}_{CMM}=\frac{\overline{X}}{\overline{X}-X_{(n)}}
\end{equation}
and
\begin{equation}
    \widehat{\theta}_{CMM}=-\widehat{\beta}_{CMM}\cdot X_{(n)}.
\end{equation}

Interested reader can refer to \cite{VILLASENORALVA20093835} and \cite{Chen02102017} to see the efficiency of these tests. It is visible that AML and CMM estimators are easy to use since these work for all values of $k$ and computationally easier also. Therefore, we use AML estimator to estimate $\theta$ and $\beta$ for the positive shape parameter case and CMM estimator for negative shape parameter.

\subsection{Power of the test}
This section presents the outcomes of a Monte Carlo simulation experiment aimed at evaluating the power of the proposed test for the GPD. Since exponential and uniform distributions are specific instances of GPD with parameters $\beta=0$ and $\beta=-1$, respectively, we evaluate the statistical power of our test for these distributions against other alternatives for a significance level $\alpha=0.05$ and for sample sizes of $n=20,\, 30\text{ and } 50$.  The power values are high for numerous aforementioned options, even with a small sample sizes and power increases as the sample size increase. We generated 1,000 samples from each alternate distribution and employed our test.  The power values are presented in Table \ref{tab:power_cases}.
\begin{table}[H]
\centering
\begin{tabular}{@{}l*{6}{m{1.5cm}}@{}}
\toprule
{Distribution} & \multicolumn{3}{c}{Exponential case} & \multicolumn{3}{c}{Uniform case} \\
\cmidrule{1-1}\cmidrule(lr){2-4}\cmidrule(lr){5-7}
 $(n,\beta)$& {$(20,0)$} & {$(30,0)$} & {$(50,0)$} & {$(20,-1)$} & {$(30,-1)$} & {$(50,-1)$} \\
\midrule
$\text{Beta}(5,5)$         & 0.999 & 1.000 & 1.000 & 0.602 & 0.820 & 0.925 \\
$\text{Weibull}(2,1)$      & 0.660 & 0.930 & 1.000 & 0.519 & 0.676 & 0.791 \\
$\text{Weibull}(3,1)$      & 0.981 & 1.000 & 1.000 & 0.654 & 0.865 & 0.938 \\
$\text{Gamma}(5,1)$        & 0.725 & 0.983 & 1.000 & 0.835 & 0.967 & 0.993 \\
$\text{Gamma}(8,1)$        & 0.944 & 1.000 & 1.000 & 0.891 & 0.977 & 1.000 \\
$\text{Gen-Gamma}(2,1/3)$  & 1.000 & 1.000 & 1.000 & 0.627 & 0.843 & 0.935 \\
$\text{Gen-Gamma}(2,1/2)$  & 0.952 & 1.000 & 1.000 & 0.774 & 0.940 & 0.990 \\
$\text{Gen-Gamma}(1,1/2)$  & 0.621 & 0.929 & 1.000 & 0.526 & 0.639 & 0.805 \\
abs$(N(2,1))$       & 0.833 & 0.985 & 1.000 & 0.369 & 0.473 & 0.602 \\
abs$(N(3,1))$       & 0.993 & 1.000 & 1.000 & 0.604 & 0.838 & 0.938 \\
$\chi^2(6)$                & 0.462 & 0.759 & 0.994 & 0.659 & 0.832 & 0.917 \\
$\chi^2(15)$               & 0.925 & 1.000 & 1.000 & 0.875 & 0.979 & 0.996 \\
abs$(\text{Gumbel}(3,2))$     & 0.375 & 0.695 & 0.963 & 0.479 & 0.641 & 0.766 \\
abs$(\text{Gumbel}(5,2))$     & 0.785 & 0.990 & 1.000 & 0.911 & 0.987 & 0.999 \\
\bottomrule
\end{tabular}
\caption{Power Analysis for Exponential ($\beta=0$) and Uniform Cases ($\beta=-1$) for a significance level $\alpha=0.05$}
\label{tab:power_cases}
\end{table}
    To estimate the power of the proposed test, we conduct simulations based on the following alternative:  Beta($\alpha_1,\alpha_2$), Weibull($\alpha_1,\alpha_2$), Gamma($\alpha_1,\alpha_2$), Generalized gamma($\alpha_1,\alpha_2$) with a positive power $\alpha_2$ of gamma variable with shape parameter $\alpha_1$, absolute value of Normal($\mu,\sigma$), Chi-square($\nu$), absolute value of Gumbel($\alpha_1,\alpha_2$).  The findings are presented in Table \ref{tab:power_analysis}.  It is important to observe that as the alternative hypothesis moves further away from the null hypothesis, or as the sample size increases, the power of the test also increases. We conduct a comparison of our test with the one proposed by \cite{VILLASENORALVA20093835}, as they introduced the estimators of $\beta$ and $\theta$, which are utilized in our test.  Upon examining [Table 2: \cite{VILLASENORALVA20093835}], it is evident that our test demonstrates greater power compared to theirs across the majority of alternative distributions. 
\begin{sidewaystable}
\centering
\resizebox{\textwidth}{!}{
\begin{tabular}{@{}l*{8}{m{2cm}}@{}}
\toprule
Distribution / $(n,\beta)$ & $(20,0.1)$ & $(20,0.2)$ & $(20,1)$ & $(30,0.1)$ & $(30,0.2)$ & $(50,0.5)$ & $(30,-0.5)$ & $(50,-0.5)$ \\
\midrule
$\text{Beta}(1,2)$ & 0.439 & 0.532 & 0.994 & 0.754 & 0.856 & 0.999 & 0.044 & 0.050 \\
$\text{Beta}(2,1)$ & 1.000 & 1.000 & 1.000 & 1.000 & 1.000 & 1.000 & 0.002 & 0.010 \\
$\text{Beta}(5,5)$ & 1.000 & 1.000 & 1.000 & 1.000 & 1.000 & 1.000 & 0.494 & 0.950 \\
$\text{Weibull}(2,1)$ & 0.732 & 0.860 & 1.000 & 0.986 & 0.997 & 1.000 & 0.423 & 0.798 \\
$\text{Weibull}(3,1)$ & 0.994 & 1.000 & 1.000 & 1.000 & 1.000 & 1.000 & 0.612 & 0.955 \\
$\text{Gamma}(5,1)$ & 0.854 & 0.937 & 1.000 & 0.999 & 1.000 & 1.000 & 0.856 & 0.995 \\
$\text{Gamma}(8,1)$ & 0.973 & 0.991 & 1.000 & 1.000 & 1.000 & 1.000 & 0.879 & 0.999 \\
$\text{Gen-Gamma}(2,1/3)$ & 1.000 & 1.000 & 1.000 & 1.000 & 1.000 & 1.000 & 0.543 & 0.935 \\
$\text{Gen-Gamma}(2,1/2)$ & 0.986 & 0.994 & 1.000 & 1.000 & 2.000 & 1.000 & 0.759 & 0.989 \\
$\text{Gen-Gamma}(1,1/2)$ & 0.726 & 0.859 & 1.000 & 0.979 & 0.997 & 1.000 & 0.480 & 0.804 \\
abs$(N(2,2))$ & 0.432 & 0.531 & 0.995 & 0.758 & 0.863 & 1.000 & 0.048 & 0.069 \\
abs$(N(2,1))$ & 0.893 & 0.954 & 1.000 & 0.998 & 1.000 & 1.000 & 0.236 & 0.579 \\
abs$(N(3,1))$ & 0.998 & 0.999 & 1.000 & 1.000 & 1.000 & 1.000 & 0.568 & 0.940 \\
$\chi^2(6)$ & 0.544 &)0.701 & 1.000 & 0.927 & 0.984 & 1.000 & 0.615 & 0.925 \\
$\chi^2(15)$ & 0.976 & 0.999 & 1.000 & 1.000 & 1.000 & 1.000 & 0.884 & 0.997 \\
abs$(\text{Gumbel}(3,2))$ & 0.506 & 0.581 & 0.999 & 0.865 & 0.949 & 1.000 & 0.415 & 0.761 \\
abs$(\text{Gumbel}(5,2))$ & 0.879 & 0.950 & 1.000 & 1.000 & 1.000 & 1.000 & 0.920 & 1.000 \\
abs$(\text{Gumbel}(5,5))$ & 0.240 & 0.301 & 0.981 & 0.465 & 0.638 & 0.998 & 0.066 & 0.131 \\
\bottomrule
\end{tabular}
}
\caption{Power of the proposed test against various alternatives for a significance level $\alpha=0.05$}
\label{tab:power_analysis}
\end{sidewaystable}
\subsection{Real life applications}
This section presents two actual datasets from real-world scenarios. While testing on real datasets, we reject the null hypothesis $H_0$, when both the alternative hypothesis for positive and negative beta, $H_1^P$ and $H_1^P$, cannot be rejected at a significance level $\alpha$. It has been seen by many authors including \cite{Hosking01081987} and \cite{zhang2009new}, that the values of shape parameter $\beta$ will be negative and specifically $-0.5\le \beta<0$, so it is better to test the negative case first.

The analysis of ozone levels in Delhi, India is presented in this section, contributing new data to the existing literature.  The second dataset we included is Bilbao waves data. It has already been examined in the literature by numerous authors for purposes such as parameter estimation or testing goodness of fit.  
\subsection*{Data of ozone (O$_3$) level in Delhi, India}
The examination of the probabilistic behavior of air pollutant concentrations in the atmosphere holds significant importance for implementing measures that promote human health protection in urban areas. Ozone is a gas consisting of three oxygen atoms.  In the higher strata of the Earth's atmosphere, it absorbs detrimental UV rays. At ground level, ozone is produced through a chemical process involving sunlight and organic gases, as well as nitrogen oxides released by automobiles, power stations, chemical facilities, and other sources. Ozone concentrations are typically elevated throughout spring and summer, whereas they are diminished in winter.  Ozone concentrations peak in the afternoon and are often greater in rural areas than in urban locales.  Ozone constitutes a significant element of summer air pollution events. Studying ozone (O$_3$) levels exceeding 100 $\mu g/m^3$ (daily maximum 8-h mean) is crucial, as they are deemed hazardous to human health based on WHO (World Health Organization) guidelines.
\begin{table}[H]
\centering
\footnotesize
\begin{tabular}{m{2.5cm}  m{1.2cm} m{2.5cm}m{1.2cm}  m{2.5cm} m{1.2cm}}
\toprule
\textbf{Date} & \textbf{Excess} & \textbf{Date} & \textbf{Excess} & \textbf{Date} & \textbf{Excess} \\
\midrule
11-Jun-2015 & 38.5 & 18-Jan-2016 & 7.94 & 22-May-2016 & 1.65 \\
20-Oct-2015 & 1.74 & 19-Jan-2016 & 4.61 & 25-May-2016 & 26.55 \\
02-Nov-2015 & 36.67 & 20-Jan-2016 & 0.12 & 26-May-2016 & 60.06 \\
04-Nov-2015 & 14.43 & 22-Jan-2016 & 36.52 & 27-May-2016 & 69.35 \\
07-Nov-2015 & 7.69 & 23-Jan-2016 & 47.56 & 28-May-2016 & 20.76 \\
08-Nov-2015 & 3.4 & 24-Jan-2016 & 32.54 & 29-May-2016 & 1.26 \\
09-Nov-2015 & 8.12 & 25-Jan-2016 & 3.53 & 03-Jun-2016 & 29.18 \\
10-Nov-2015 & 15.45 & 26-Jan-2016 & 39.27 & 04-Jun-2016 & 15.31 \\
13-Nov-2015 & 1.14 & 27-Jan-2016 & 42.54 & 06-Jun-2016 & 0.24 \\
15-Nov-2015 & 1.15 & 28-Jan-2016 & 51.6 & 24-Jun-2016 & 1.3 \\
19-Nov-2015 & 1.68 & 29-Jan-2016 & 75.04 & 20-Sep-2016 & 34.72 \\
20-Nov-2015 & 24.74 & 30-Jan-2016 & 57.74 & 21-Sep-2016 & 44.45 \\
21-Nov-2015 & 14.3 & 05-Feb-2016 & 21.22 & 28-Sep-2016 & 13.81 \\
22-Nov-2015 & 16.29 & 06-Feb-2016 & 49.67 & 01-Oct-2016 & 9.06 \\
23-Nov-2015 & 32.26 & 09-Feb-2016 & 6.92 & 04-Oct-2016 & 23.99 \\
24-Nov-2015 & 13.35 & 10-Feb-2016 & 29.76 & 06-Oct-2016 & 9.08 \\
30-Nov-2015 & 31.34 & 12-Feb-2016 & 17.2 & 09-Oct-2016 & 28.17 \\
05-Dec-2015 & 29.79 & 13-Feb-2016 & 1.47 & 10-Oct-2016 & 12.52 \\
06-Dec-2015 & 22.22 & 19-Feb-2016 & 8.89 & 11-Oct-2016 & 4.67 \\
07-Dec-2015 & 3.29 & 24-Feb-2016 & 100.41 & 15-Oct-2016 & 20.33 \\
08-Dec-2015 & 39.79 & 26-Feb-2016 & 15.84 & 17-Oct-2016 & 12.83 \\
09-Dec-2015 & 31.99 & 27-Feb-2016 & 17.3 & 18-Oct-2016 & 1 \\
10-Dec-2015 & 23.11 & 28-Feb-2016 & 30.99 & 20-Oct-2016 & 7.03 \\
11-Dec-2015 & 19.41 & 29-Feb-2016 & 37.57 & 23-Oct-2016 & 22.74 \\
12-Dec-2015 & 22.43 & 01-Mar-2016 & 38.3 & 24-Oct-2016 & 9.22 \\
13-Dec-2015 & 1.58 & 02-Mar-2016 & 57.07 & 28-Oct-2016 & 7.12 \\
23-Dec-2015 & 11.69 & 03-Mar-2016 & 13.89 & 29-Oct-2016 & 39.22 \\
01-Jan-2016 & 8.14 & 04-Mar-2016 & 3.88 & 30-Oct-2016 & 157.73 \\
05-Jan-2016 & 31.4 & 24-Mar-2016 & 69.36 & 31-Oct-2016 & 53.52 \\
06-Jan-2016 & 35.04 & 25-Mar-2016 & 42.38 & 01-Nov-2016 & 93.31 \\
07-Jan-2016 & 86.07 & 21-Apr-2016 & 3.07 & 02-Nov-2016 & 60.01 \\
08-Jan-2016 & 50.8 & 22-Apr-2016 & 10.09 & 03-Nov-2016 & 8.33 \\
09-Jan-2016 & 0.54 & 28-Apr-2016 & 17.59 & 04-Nov-2016 & 77.07 \\
10-Jan-2016 & 30.02 & 03-May-2016 & 1.53 & 25-Oct-2017 & 6.16 \\
11-Jan-2016 & 40.72 & 18-May-2016 & 5.42 & 28-Oct-2017 & 10.71 \\
12-Jan-2016 & 31.63 & 21-May-2016 & 7.9 & 10-Nov-2017 & 4.94 \\
\bottomrule
\end{tabular}
\caption{Ozone level excess data (in $\mu g$/$m^3$) of Delhi, India for June 2015 to November 2017}
\label{ozone_data_delhi}
\end{table}

We use the data of ozone level where the ozone level cross the limit 100 $\mu g/m^3$, it is presented in Table \ref{ozone_data_delhi}. This data have been registered during June 2015 to November 2017 in an air quality monitoring station in Delhi, India. The data has been made publicly available by the ``Central Pollution Control Board: https://cpcb.nic.in/'' which is the official portal of Government of India.

We implement our proposed test on this dataset with $n=108$, resulting in $\widehat\Delta_N^*=0.0232$ and $\beta_{CMM}=-0.1955$.  The critical value for \( n=108 \) and \( \beta=-0.1955 \) is \( 0.0342 \) for a 99\% confidence level.  The value of $\widehat\Delta_N^*$ is lower than the crucial value, hence, we do not reject the null hypothesis.  Therefore, there is no evidence contradicting the hypothesis that these data follow to a generalized Pareto distribution with a shape parameter $\beta<0$. We checked our test for positive $\beta$ case too, where we get $\widehat{\Delta}_P=0.1663$ and $\beta_{AML}=0.4251$. The critical value for \( n=108 \) and \( \beta=0.4251 \) is \( 0.1391 \) for a $0.01$ significance level, so we reject the null hypothesis for positive shape parameter value. We verify this data using Kolmogorov-Smirnov (K-S) test, Anderson-Darling and Chi-Square test, the test statistic values for these tests are $0.04528$, $0.33678$ and $2.3751$, respectively. None of these three tests reject the null hypotheses when a negative shape parameter value is present. This strengthens our claim. 

\subsection*{Bilbao Waves Data}
The example comprises the zero-crossing hourly mean periods (in seconds) of the sea waves recorded at the Bilbao buoy, Spain.  The data serve to examine the impact of periods on beach morphodynamics and other characteristics associated with the right tail studied by \cite{Castillo01121997}.  Table \ref{tab:bilbao_waves} includes only data exceeding 7 seconds. 
\begin{table}[H]
\centering
\footnotesize
\begin{tabular}{*{14}{c}}
\toprule
7.05 & 7.26 & 7.46 & 7.59 & 7.69 & 7.82 & 7.90 & 7.97 & 8.11 & 8.21 & 8.40 & 8.51 & 8.69 & 8.85 \\
9.06 & 9.23 & 9.46 & 9.75 & 9.12 & 9.24 & 9.47 & 9.78 & 9.16 & 9.27 & 9.59 & 9.79 & 9.43 & 9.74 \\
7.12 & 7.27 & 7.46 & 7.59 & 7.72 & 7.83 & 7.91 & 7.99 & 8.12 & 8.23 & 8.41 & 8.52 & 8.71 & 8.86 \\
7.15 & 7.28 & 7.47 & 7.61 & 7.72 & 7.83 & 7.93 & 8.00 & 8.15 & 8.23 & 8.42 & 8.53 & 8.72 & 8.88 \\
7.18 & 7.30 & 7.48 & 7.63 & 7.72 & 7.83 & 7.93 & 8.03 & 8.15 & 8.30 & 8.43 & 8.54 & 8.74 & 8.88 \\
9.17 & 9.29 & 9.59 & 9.79 & 9.17 & 9.30 & 9.60 & 9.80 & 9.18 & 9.32 & 9.61 & 9.84 & 9.22 & 9.90 \\
7.19 & 7.31 & 7.48 & 7.65 & 7.72 & 7.84 & 7.93 & 8.03 & 8.15 & 8.30 & 8.43 & 8.56 & 8.74 & 8.94 \\
7.20 & 7.31 & 7.52 & 7.66 & 7.72 & 7.85 & 7.94 & 8.05 & 8.18 & 8.31 & 8.45 & 8.58 & 8.74 & 8.98 \\
7.20 & 7.32 & 7.54 & 7.66 & 7.77 & 7.85 & 7.95 & 8.06 & 8.18 & 8.31 & 8.48 & 8.59 & 8.74 & 8.98 \\
7.20 & 7.33 & 7.55 & 7.67 & 7.77 & 7.88 & 7.95 & 8.06 & 8.18 & 8.32 & 8.49 & 8.59 & 8.79 & 8.99 \\
7.20 & 7.37 & 7.55 & 7.67 & 7.79 & 7.88 & 7.97 & 8.07 & 8.19 & 8.32 & 8.50 & 8.60 & 8.81 & 9.01 \\
7.25 & 7.40 & 7.58 & 7.68 & 7.79 & 7.90 & 7.97 & 8.10 & 8.20 & 8.33 & 8.50 & 8.65 & 8.84 & 9.03 \\
9.18 & 9.33 & 9.62 & 9.85 & 9.18 & 9.36 & 9.63 & 9.89 & 9.21 & 9.38 & 9.66 & & & \\
\bottomrule
\end{tabular}
\caption{The Bilbao waves data: the zero-crossing hourly mean periods (in seconds), above 7 sec, of the sea waves measured in Bilbao buoy.}
\label{tab:bilbao_waves}
\end{table}

The application of GPD to this dataset has been extensively examined in existing literature (e.g., \cite{Castillo01121997}, \cite{zhang2009new}).  It has been observed by \cite{zhang2009new}, that when the threshold time $t\ge 7.5$, the Generalized Pareto Distribution effectively models the exceedance. We apply our test to the data when $t\ge 7.5$, then we get $n=154$, the estimated shape parameter value and test statistics as $\widehat{\beta}_{CMM}=-0.7133$ and $\widehat\Delta_N^*=0.01208$ respectively. The value of $\widehat{\beta}_{CMM}$ is close to many other good estimators of $\beta$ given in \cite{Chen02102017}. The p-value for this data is $0.291$ for the negative shape parameter case. Therefore, we can not reject the null hypothesis $H_0$ at a significance level $0.05$.

\section{Conclusion}
A characterization for GPD is presented utilizing Stein's type identity, even though the support of GPD varies according to the shape parameter $\beta$.  Characterization for univariate distributions with either semi-bounded or bounded support has been introduced by \cite{betsch2021fixed}. We also present an alternative characterization of GPD through the lens of dynamic survival extropy, which serves as a characterization for exponential, uniform, and modified GPD at specific values of the proportionality constant $k$. We independently present a goodness of fit test tailored for both positive and negative values of $\beta$. For $\beta>0$, we employ a characterization based on Stein's identity, while for $\beta<0$, we utilize a characterization grounded in dynamic survival extropy.  Our test offers a more straightforward and simple calculation than traditional methods, including the Kolmogorov-Smirnov test and the Anderson-Darling test.  A Monte Carlo simulation study utilizing various sample sizes and shape parameter values demonstrates that it maintains high power, even with small sample sizes. The asymptotic properties of the test have been established on the premise that a consistent estimator of $\theta$ and $\beta$ will be utilized for the evaluation of test statistics. Recognizing the inherent challenge posed by censored data, we expanded the test to accommodate right censored data. Real-life applications utilizing actual datasets have been incorporated, including the ozone (O$_3$) level data exceeding 100$\mu g/m^3$ in Delhi, India, from June 2015 to November 2017 and zero crossing hourly mean period (in seconds), above 7 sec, of the sea waves in Bilbao buoy, Spain. Our proposed test is utilized to determine if the excess data follows the Generalized Pareto Distribution or not.
\section*{Conflict of interest}
No conflicts of interest are disclosed by the authors.
\section*{Funding} 
Gaurav Kandpal would like to acknowledge financial support from the University Grant Commission, Government of India (Student ID : 221610071232).

\printbibliography

\end{document}